\documentclass{eptcs}
\usepackage{times}
\usepackage{amsmath,amssymb}
\usepackage{paralist}
\usepackage{comment}
\usepackage{jeffe}
\usepackage{enumitem}
\usepackage{algorithm}
\usepackage{graphicx}
\usepackage[noend]{algorithmic} 
\usepackage[bottom]{footmisc}

\newcommand{\red}[1]{{\textcolor{red}{#1}}}
\newcommand{\todo}[1]{\noindent\textbf{TODO: }\marginpar{****}%
\textit{{\red{#1}}}\textbf{ :ODOT}}

\newcommand{\G}{\mathcal{G}}

\renewcommand{\u}{\mathsf{u}}

\newcommand{\A}{\mathcal{A}}

\newcommand{\lead}{\textsc{leader}}
\newcommand{\leader}{\textsc{leader}}
\newcommand{\beep}{\textsc{beep}}

\renewcommand{\ge}{\geqslant}
\renewcommand{\le}{\leqslant}
\renewcommand{\geq}{\geqslant}
\renewcommand{\leq}{\leqslant}

\newtheorem{theorem}{Theorem}
\newtheorem{lemma}{Lemma}
\newtheorem{obs}{Observation}

\newboolean{short}
\setboolean{short}{true}

\newcommand{\shortOnly}[1]{\ifthenelse{\boolean{short}}{#1}{}}
\newcommand{\onlyShort}[1]{\ifthenelse{\boolean{short}}{#1}{}}
\newcommand{\longOnly}[1]{\ifthenelse{\boolean{short}}{}{#1}}
\newcommand{\onlyLong}[1]{\ifthenelse{\boolean{short}}{}{#1}}

\begin{document}
\title{Robust Leader Election in a Fast-Changing World}
\author{John Augustine\thanks{Department of Computer Science and Engineering, Indian Institute of Technology Madras, Chennai, India.  \hbox{E-mail}:~{\tt augustine@cse.iitm.ac.in, tejasvijaykulkarni@gmail.com, paresh.nakhe@gmail.com}.}  
\and Tejas Kulkarni$^\ast$
\and Paresh Nakhe$^\ast$
 \and Peter Robinson\thanks{Division of Mathematical
Sciences, Nanyang Technological University, Singapore 637371.
\hbox{E-mail}:~{\tt peter.robinson@ntu.edu.sg}}}

\newcommand{\authorrunning}{J.\ Augustine, T.\ Kulkarni, P.\ Nakhe, P.\ Robinson}
\newcommand{\titlerunning}{Robust Leader Election in a Fast-Changing World}

\maketitle

\begin{abstract}
We consider the problem of electing a leader among nodes in a highly dynamic
network where the adversary has unbounded capacity to insert and remove nodes (including the leader) from the network and change connectivity at will. 
We present a randomized algorithm that (re)elects a leader
in $O(D\log n)$ rounds with high probability, where $D$ is a bound on the
dynamic diameter of the network and $n$ is the maximum number of nodes in the network at any point in time.
We assume a model of broadcast-based communication where a node can
send only $1$ message of $O(\log n)$ bits per round and is not aware of the
receivers in advance.
Thus, our results also apply to mobile wireless ad-hoc networks,
improving over the optimal (for deterministic algorithms) $O(Dn)$ solution
presented at FOMC 2011. We show that our algorithm is optimal by  proving that {\em any} randomized Las Vegas algorithm takes at least $\Omega(D\log n)$
rounds to elect a leader with high probability,
which shows that our algorithm yields the best possible (up to constants)
termination time.
\end{abstract}



\section{Introduction}\label{sec:intro}

Electing a leader among distributed nodes in a dynamic environment is a 
fundamental but challenging task.
Protocols that were developed for static communication networks (cf.\ 
\cite{Lyn96} and references therein) are not applicable in settings 
where the network topology is continuously evolving due to mobility and 
failures of nodes.
Real world dynamic networks are continuously evolving, which requires
algorithms to function even when the network itself is constantly changing.

In this work, we consider the Dynamic Leader Election problem (cf.\
Section~\ref{sec:DLE}), which was introduced as ``Regional Consecutive Leader
Election'' in \cite{CRW:2010}.
Intuitively speaking, we consider a network of
(not necessarily fixed) distributed nodes that are required to elect a unique leader
among themselves within some bounded time.
Since any node (including the leader) can leave the network --- due to 
being out of reach from any other node or by simply crashing --- we 
require nodes to detect the absence of the leader and reelect a new leader 
as soon as possible.

Dynamic Leader Election is clearly impossible in absence of any guarantees 
on information propagation between nodes.
Similarly to \cite{CRW:2010,CRW:2011}, we assume a communication diameter 
$D$ which bounds the time needed for information propagation between nodes 
in the network.
Intuitively speaking, if some node $u$ sends some information $I$ in round 
$r$ and some node $v$ remains in the network during rounds $[r,r+D]$, then 
$v$ is guaranteed to receive the information of $u$.
Note that, in contrast to the notion dynamic diameter in undirected 
dynamic networks (cf.\ \cite{KMO:11}), our communication diameter 
assumption does not give symmetric guarantees for communication between 
nodes.
In fact, by the time that $v$ receives $I$, node $u$ might have long left 
the network and thus is unable to receive any information sent by $v$.

In this work, we present a randomized algorithm that provides fast termination
bounds even though the dynamic network is designed by an \emph{oblivious adversary} who knows our algorithm, but is unaware of the outcome of the
private coin flips. In essence, the oblivious adversary must commit in
advance to all changes made to the dynamic network before the algorithm starts executing.
Assuming
an {oblivious adversary} is a suitable choice for worst case analysis when the dynamism of network is not under control of a malicious force that is capable of inferring the current state of the nodes in the network but rather 
caused by mobility and hardware limitations of the participating nodes.

Our system model (cf.\ Section~\ref{sec:model}) requires nodes to send messages
by broadcasting them to their (current) neighbors in the network.
Moreover, the sender of a message does \emph{not} know its current neighbors.
Considering that the adversary can dynamically change the topology of the
network and the participating nodes themselves, our assumptions are suitable for
modeling real world dynamic networks, including mobile wireless ad-hoc networks
(MANETS).
In general, care must be taken to account for message loss due to collisions of
wireless broadcasts.
Similarly to \cite{CRW:2011}, we focus on the algorithmic aspects of the 
dynamic network model where neighboring nodes can communicate reliably in 
synchronous rounds, which can be guaranteed by using the Abstract MAC 
Layer (cf.\ \cite{KLN09}).

\subsection{Our Main Results:}
We present asymptotically optimal results for solving Dynamic Leader 
Election when nodes have access to unbiased private coin flips.   Our main contributions are the following:
\begin{compactenum}
\item We present a randomized algorithm that guarantees a termination 
bound of $O(D\log n)$ where $n$ is the maximum number of nodes in the 
network at any time (cf.\ Theorem~\ref{thm:main}). This improves over the 
deterministic $O(D n)$ bound of \cite{CRW:2011}. This result holds even when the dynamic network is designed by an oblivious adversary that knows the algorithm and completely controls the network dynamics, but must commit 
to its choices in advance without any knowledge of the private random bits used to execute the algorithm.
\item We show that any randomized algorithm takes at least $\Omega(D\log 
n)$ rounds to terminate with high probability\footnote{Throughout this 
  paper, ``with high probability'' means with probability at least $1-O(\frac{1}{n})$.} 
(cf.\ Theorem~\ref{thm:lowerbound}). Interestingly, this lower bound holds even if the oblivious adversary does \emph{not} have any
knowledge of the actual algorithm in place.

\end{compactenum}

\subsection{Other Related Work} \label{sec:related}
Existing solutions to the Dynamic Leader Election problem
\cite{CRW:2010,CRW:2011} consider only deterministic algorithms, albeit for the
more powerful omniscient adversary that knows the entire execution of the
algorithm in advance.

In more detail, \cite{CRW:2011} proves a lower bound that shows that any 
algorithm takes at least $\Omega(D n)$ rounds for termination in the worst 
case, in the presence of an omniscient adversary.
\cite{CRW:2011} also presents a matching (deterministic) algorithm that 
guarantees termination in $O(D n)$ time and improves over the algorithm of 
\cite{CRW:2010} by only sending one broadcast per round.

In the context of mobile ad-hoc networks, \cite{MWV00} presents an 
algorithm for electing a leader in each connected component of a network with a
changing topology and proves its correctness for the case when there is a single 
topology change. In \cite{ISWW09}, a leader election protocol for 
asynchronous networks with dynamically changing topologies is described 
that elects a leader as soon as the topology becomes stable.
Several other leader election algorithms for mobile environments are considered
in \cite{BA06,MWV00,ISWW09,MAB06,PKY04,VKT04,DLP:10}.

\cite{HPSTT99} presents a leader election algorithm in a model where the 
entire space accessible by mobile nodes is divided into non-intersecting 
subspaces.
Upon nodes meeting in a common subspace, they decide on which node 
continues to participate in the leader election protocol.
Moreover, \cite{HPSTT99} presents a probabilistic analysis for the case 
where the movement of nodes are modeled as random walks.

The dynamic network model of \cite{APRU12} is similar to the model used in 
this paper, in the sense that nodes can leave and join the (peer-to-peer) 
network over time and the topology of the network can undergo some 
changes.
\cite{APRU12} show how to solve almost everywhere agreement despite high 
amount of node churn, which in turn can be used to enable nodes to agree on a
leader.
However, \cite{APRU12} assumes that the underlying network topology 
remains an expander which does not necessarily apply to mobile ad-hoc networks.

There are established lower bounds on time and message complexity for 
leader election algorithms (cf.\ 
\cite{KorachPODC1984,AG1991:SICOMP,KPPRT2013:PODC}) in {static} synchronous networks.
These lower bounds do not apply to our system model and the Dynamic Leader 
Election, due to the fundamental impact of our additional assumptions, 
i.e.,  nodes leaving and joining the network and changes to the 
communication topology.

Similar to the technique that we use in our  leader election algorithm, the work of \cite{atish} uses exponential random variables in the context of finding dense subgraphs on dynamic networks.

The paper is organized as follows. In Section~\ref{sec:model}, we describe the dynamic network model that we study and formulate the dynamic leader election problem. We show that any randomized algorithm takes at least $\Omega(D\log 
n)$ rounds to terminate with high probability in Section~\ref{sec:lower}. To complement this negative result, we provide an optimal algorithm that terminates in $O(D \log n)$ rounds with high probability in Section~\ref{sec:upper}.


\section{Preliminaries}\label{sec:model}

\subsection{\textbf{System Model}} We consider a
dynamic network of nodes modelled as a dynamically changing infinite sequence of
graphs $\G = (G^{1}, G^2, \ldots)$. Each node runs an instance of a distributed
algorithm and computation is structured into synchronous rounds. 
We assume that all nodes have access to a common global clock.
Nodes communicate with their neighbors via broadcast communication where
a node  is restricted to broadcast at most $1$ message of $O(\log n)$
bits per round.
The graph $G^r = (V^r, E^r)$, $r \ge 1$,
represents the state of the network in round $r$, i.e., the vertex set
$V^r$ is the set of nodes in the network and  the edges $E^r$ represent
the  connectivity in the network. The
maximum number of nodes in the network at any time is denoted by $n$, i.e., $\forall $r$, |V^r| \le n$. 
For two rounds $r_a$ and $r_b$ where $r_a \leq r_b$ we define 
$V^{[r_a,r_b]} = \{ v \mid \forall k \in [ r_a , r_b], v \in V^k \}$.    
We assume that each node enters the network once and may therefore leave the network at most once. A node that enters the network in round $r$ and leaves the network at the end of round $r'$  ($ = \infty$ if the node never leaves) will be in $V^{[r,r']}$.
If $e = (u,v) \in E^r$, then any message
sent in round $r$ by $u$ is received by $v$ in round $r$ and vice versa. 
Since communication is broadcast based, node $v$ does not know that $u$ is its in round $r$ until it has received a message from $u$.
In particular, when node $u$ broadcasts a message in round $r$, it does
not know which nodes will receive this message.

The changes in node and edge sets are made by an adversary that is
\emph{oblivious} to the state of the nodes. In other words,  the adversary
commits to the sequence $\G$ before round 1 after which it cannot make
anymore changes. (Nodes themselves are not aware of this sequence in
advance.) 
Thus, we can view each round as a sequence of three events:
\begin{compactitem}
\item The network is updated to $G^r$ (fixed in advance by the adversary).
  We say that a node \emph{survives until $r$} if it remains in $V^r$.
\item Nodes perform local computation including (private) random coin flips.
\item Nodes communicate with their neighbors.
\end{compactitem}
An \emph{execution} of an algorithm is entirely determined by a sequence
$\G$ and the outcome of the local coin flips.

So far, we have not imposed any restrictions on how the adversary can
modify the node and edge sets. In fact, we do not even require that any particular graph in $\G$ be connected.
However, we now introduce a weak restriction on the propagation of information
among nodes.
We say that a node $u$ \emph{floods} a message $M$ starting in round $r$ if it broadcasts $M$ to its neighbors in every round starting from $r$ as long as it is alive {\em and} every other node that receives $M$ continues in turn to flood $M$.
 We assume that  if a node
$u \in V^r$ floods a message starting from  round $r$, then every node that is in
$V^{[r,r+D]}$ will receive the message  by round $r+D$.  We call $D$ 
the \emph{bounded communication diameter} of the network and assume that $D$ is common knowledge of all nodes.

As a further convenience, the model allows conditional flooding in which a message $M$ may be flooded until some condition is met. 
More precisely, a conditionally flooded message comprises of both a message and a condition. When a node receives a conditionally flooded message, it continues to flood the message if and only if the condition is true.
Such conditional flooding will be useful in ensuring that beep messages sent by a leader to indicate its presence in the network don't persist forever in the network. These beep messages can be flooded under the condition {\tt <current beep message is latest and generated in last D rounds>}, thereby implicitly discarding stale beep messages.

\subsection{\textbf{Dynamic Leader Election}}  \label{sec:DLE}
Given a synchronous dynamic
network $\G$, we want to develop an algorithm to elect a leader in bounded
time in spite of the churn. Each node $u$ is equipped with a variable $\lead_u$
which is initialized to $\bot$. A node $u$ {\em chooses} or {\em elects}  a leader node $v$ by
setting $\lead_u \leftarrow v $; note that this allows a node $u$ to elect
itself as the leader by setting $\lead_u \leftarrow u$. 
The leader election algorithm so developed should satisfy the following conditions:
\begin{description}
\item[{Agreement:}] If $u$ and $v$ are any two nodes in $G^{r}$ such that, $\leader(u) \neq \bot$ and $\leader(v) \neq \bot $ then $\leader(u) = \leader(v)$. This implies that at any one round, there is at most one leader in the network.
\item[Termination:] Any node $u$ without a leader should elect a leader
  within a bounded number of rounds; formally, this corresponds to assigning $\lead_u$. We say that an algorithm has {\em termination time}
  $T$ if every node $u$ that has $\lead_u=\bot$ in some round $r$, has set
  $\lead_u\neq\bot$ by some round $r'\le r+ T$, assuming that $u \in V^{[r,r+T]}$.
\item[Validity:] If some node $u$ elects distinct node $v$ as its leader in round $r$, then $v$ must have been the leader in some round in $[r - D-1, r]$.
\item[Stability:]  If a node $u$ stops considering a node $v$ as its leader, then $v$ has left the network.
\end{description} 

We are interested in tight bounds on the termination time $T$. From~\cite{CRW:2011}, we know that $T \in \Omega(nD)$ when we require deterministic bounds on $T$. Our interest, however, is to attain termination times that are significantly lower. Towards this goal, we focus on randomized algorithm that provide significantly better bounds on $T$ that hold with high probability of the form $1 - 1/n^{\Omega(1)}$.


\section{Lower Bound on Termination Time}\label{sec:lower}
In this section, we wish to show a lower bound on the termination time $T$ 
of any randomized algorithm. More precisely, we wish to show that, for any 
algorithm $\A$ that solves dynamic leader election (with probability $1$), the number of rounds until $\A$ terminates with high 
probability is at least $\Omega(D \log n)$, where $n$ is the most number of nodes 
in the network at any given time. Recall that we have assumed that our 
adversary is oblivious in that it is aware of the algorithm used, but is 
oblivious to the outcomes of coin flips used by the algorithm. In effect, 
we can assume that the adversary generates $\G$ before the algorithm 
starts its execution. In particular, for the purpose of proving the lower 
bound, the adversary follows the strategy described in Algorithm~\ref{alg:adv}.  

\begin{algorithm}
\caption{Strategy of the Oblivious Adversary.}
\label{alg:adv}               
\begin{algorithmic}[1]
\STATE \COMMENT{This algorithm produces a sequence of graphs that make up 
  an input instance to the Dynamic Leader Election problem.  
The adversary equips each node with a unique id chosen uniformly at random from the set $\{1,\dots,n^5\}$; once an id has been used for a node, it is removed from this set.}
\vspace{0.05in}
\STATE We begin with $n$ nodes in the network. 
\FOR {$i = 0,1,2...$}
\vspace{0.05in}
\STATE \COMMENT{Generating $G^{iD+1}$ to $G^{(i+1)D -1}$:}
\STATE From round $iD+1$ to $(i+1)D-1$, all nodes are pairwise disconnected. 
\vspace{0.05in}
\STATE \COMMENT{Generating $G^{(i+1)D}$:}
\STATE In round $(i+1)D$, each node is removed with probability $1/2$. Surviving nodes become a part of $G^{(i+1)D}$.
\STATE New nodes are added to bring the cardinality of nodes in $G^{(i+1)D}$ back to $n$.
\STATE The $n$ nodes form a completely connected network in round  $G^{(i+1)D}$.
\ENDFOR
\end{algorithmic}
\end{algorithm}

The following observation confirms that the adversarial strategy as described in
Algorithm~\ref{alg:adv} is valid according to the modeling assumptions state in
Section~\ref{sec:model}.

\begin{obs}
The number of nodes in each $G \in \G$ generated by Algorithm\ 
\ref{alg:adv}  is $n$ and remains so throughout its lifetime. Furthermore, 
$\G$ generated by Algorithm\ \ref{alg:adv}  has a communication diameter 
bounded by $D$ and forms a valid input for the dynamic leader election 
problem.
\end{obs}

We say that a node $\ell$ is an {\em effective leader} if it has executed 
$\leader(\ell) \leftarrow \ell$ and there exists at least one other node 
$u$ that  is aware that $\ell$ has elected itself leader.
If no such node $u$ exists, we say that $\ell$ is \emph{ineffective}.

The validity and agreement conditions imply the following Observation~\ref{obs:effective}: 

\begin{obs} \label{obs:effective}
Suppose that some node $\ell$ is an ineffective leader during $[r_1,r_2]$.
Then no distinct node can set its $\lead$ variable to a value $\neq \bot$ 
during $[r_1,r_2]$.
\end{obs}

\begin{lemma}\label{lem:induct}
Suppose that there is no leader at the start\footnote{The \emph{start of 
    round $r+1$} and the \emph{end of round $r$} happen at the exact same 
  point in time.}  of round $iD+1$, for any $i \ge 0$. Then there will be 
no effective leader until the end of round $(i+1)D -1$.
\end{lemma}
\begin{proof}
If there is no node that elects itself leader until the end of round 
$(i+1)D -1$, we are done. Therefore, let node $\ell$ elect itself leader 
in some round in $[iD+1, (i+1)D -1]$. The adversarial strategy (cf.\ 
Algorithm~\ref{alg:adv}), however, ensures that there is no communication 
going out from $\ell$ until round $(i+1)D$. From Observation~\ref{obs:effective}, it follows that $\ell$ cannot 
become an effective leader until the end of round $(i+1)D -1$.
\end{proof}

\begin{lemma}\label{lem:prob}
Consider any $i > 0$.  If there is no effective leader up until the end of 
round $iD-1$, then, with probability at least $1/2-\varepsilon$, for any fixed constant $\varepsilon >0$, there will be no leader 
node at the end of round $iD$.
\end{lemma}
\begin{proof}
 At the end of round $iD-1$, there are two possible exhaustive cases.
In the first case, there is no leader until the end of round $iD-1$.
Since $\A$ is a Las Vegas algorithm, it follows that there can be at most $1$ id $id_1$ such that a node equipped with $id_1$ has nonzero probability of (immediately) becoming the leader upon entering in round $iD$.
(Otherwise, a simple indistinguishibility argument provides a contradiction to $\A$ achieving agreement with probability $1$.)
According to Algorithm~\ref{alg:adv}, the adversary chooses new ids uniformly at random from at least $n^4$ ids (discarding previously used ids).
Therefore, the probability of choosing $id_1$ for a new node in round $i D$ is at most $1/n^2$.
Now suppose that a node that was present in the network in round $iD-1$ elects itself as the leader in round $iD$. Since, all such nodes are removed with probability $1/2$, the lemma follows for this case.

In the second case, there is a leader $\ell$ but other nodes are unaware.
The leader $\ell$ is removed in round $iD$ with probability $1/2$ before it gets a chance to inform other nodes that it is the leader. 
\end{proof}

\begin{theorem} \label{thm:lowerbound}
Let $\A$ be a randomized Las Vegas algorithm that achieves dynamic leader election. 
The probability that $\A$ terminates in at most $iD$ rounds is at
most $1-\frac{1}{2^{2i+1}}$.
In particular, algorithm $\A$ requires $\Omega(D\log n)$ rounds to guarantee termination with high probability.
\end{theorem}
\begin{proof}
  Let $A_i$ be the event that $\bigcap_{r=0}^{iD} V^r \ne \emptyset$, i.e., some node remains in the network until round $iD$. 
  According to the adversarial strategy (cf.\ Algorithm~\ref{alg:adv}) we have Pr$[A_i]= 2^{-i}$. Let $B_i$ denote the event that there is no leader at the end of round $iD$.
  Note that, if some node $u$ has remained in the network and there is no leader in round $iD$, then, by validity and stability, $u$ has not set $\leader(u)$ at any point up to round $iD$. 
  Inductively applying Lemmas~\ref{lem:induct} and \ref{lem:prob} shows that  Pr$[B_i] \geq (1/2 - \varepsilon)^i \ge 4^{-i}$, for any $i>0$.
In other words, the probability that some node (that remained in the network) has not yet set its leader variable at any point during $[1,iD]$ is $Pr[A_i \cap B_i] \geq \frac{1}{2^{2i+1}}$.
\end{proof}


\section{Randomized Dynamic Leader Election}\label{sec:upper}

We now present a randomized algorithm for solving dynamic leader election that
terminates with high probability in $O(D \log n)$ rounds. The algorithm must be
designed to maintain a leader as long as it is in the network and elect a new
leader when the current leader leaves the network. The high level framework of
the algorithm is as follows. When an elected leader is in the network, the
leader floods out timestamped beep messages every round that are heard by all
other nodes.  When a node has heard beep messages with timestamp in the last
$D$ rounds, it assumes that the leader is still around. Otherwise, it can
conclude that the leader has left the network and  sets their leader variable
to $\bot$. A new election process then starts afresh during which the nodes
elect a new leader. Algorithm~\ref{alg:frame} presents the detailed pseudo code.   

Given the framework described above, we can focus on the leader election process which starts when  the leader has left the network and a new leader must be elected. In general, all nodes are aware of the current round number and the leader election process will start at the next round number that is of the form $2iD+1$ for some integer $i$. 
Given that our emphasis is on the leader election process, we will assume for simplicity that the leader election process starts at round number 1.  Our algorithm operates in phases with each phase consisting of  $2D$  rounds. The first phase is rounds 1 through $2D$ and the second phase is from rounds $2D+1$ through $4D$, and so on. In general, a round number of the form $2iD+1$ is the start of a new phase, and that phase ends in round $2(i+1)D$.
Since all nodes see a common clock, the round numbers in which new phases begin is common knowledge.  
  
  \begin{algorithm}
\caption{High level framework of the leader election algorithm.}    
\label{alg:frame}      
\begin{algorithmic} [1]  
\FOR {every round at every node $u$}
\IF{leader = $u$}
\STATE Generate and flood a \beep\ message every round under the condition that only the latest \beep\ message is propagated and \beep s older than $D$ rounds are discarded.
\ELSIF{the node $u$ has entered the network in the current round}
\STATE Wait passively until one of the two events occur:
\IF{for a full phase (i.e., all 2D rounds of a phase) $u$ has not heard a beep timestamped in the last $D$ rounds}
\STATE Become an active participant of the leader election process (Algorithm~\ref{alg:election}) starting from the next phase.
\ELSE[$u$ has received a beep message:]
\STATE Set leader variable to the id contained in the latest beep message.
\ENDIF

\ELSE
\IF{no \beep\ message from last $D$ rounds was heard}
\STATE Start the election process (Algorithm~\ref{alg:election}) starting from the next phase.
\ENDIF
\ENDIF
\ENDFOR
\end{algorithmic}
\end{algorithm}

In each phase, the network attempts to elect a leader. 
We say that the network is \emph{successful} in a phase if a leader  is elected (i.e., sets its own leader variable to itself) in the first $D$ rounds and  the leader stays in the network till the end of the phase. In a successful phase, the leader  announces its leadership via timestamped beep messages sent out in (every round of) the second half of the phase. This ensures that the current leader election process has terminated.
During a phase if either (i) no node elects itself as leader or (ii) if the elected leader leaves the network before  the end of the phase, then we say that the network \emph{failed} in that phase.  In a failed phase, there is no guarantee that the leader election process would have terminated.

When a new node enters the network, it starts out as a passive participant
because it will be unaware of whether a leader is present or not. If a passive
participant  hears a beep message with timestamp within last $D$ rounds, it
will appropriately  update its leader variable and cease to be passive.
Alternatively, when a passive participant has been in the network for a full
phase (i.e., all 2D rounds of some phase) without receiving such a message, it knows that a leader election process is still underway and becomes active starting from the next phase.
When a node is passive, it will not be a candidate for becoming a leader, but, nevertheless, it participates in the algorithm by forwarding messages.

At the start of a new phase, each active node $u$ draws a random number from
the exponential distribution with parameter $\lambda = 2^{p_u}$, where $p_u$ is
the number of phases spent by that node in the current election process. We
assume that the random numbers are generated with sufficient precision such
that no two generated random numbers are equal. Each node then creates a
message comprising its exponentially distributed random number\footnote{For simplicity, our description assumes that the exponential random numbers are generated by each node and the generated number is explicitly represented in the rank message. This explicit representation could blowup the number of bits in the rank messages. We note that exponential random numbers are typically generated from uniformly distributed random numbers~\cite{Dev86}. So instead of passing exponential random numbers explicitly, we can pass them implicitly by just passing the uniformly distributed random numbers, with the exponential random numbers computed at the destinations. An added advantage is that it is easy to see that $O(\log n)$ bits of the uniform random numbers is sufficient to ensure that two (implied exponential) random numbers will be equal with probability at most $1/n^{k}$ for any fixed $k$. } and its unique id; we call this 
its \emph{rank message}. 
Given two rank messages, we say that the rank message
containing the smaller random number is the smaller rank message. The algorithm
aims to elect the node that generated and flooded the smallest rank message as
the leader.  Towards this goal, each node broadcasts its rank message, but a node
will only continue to send a message $m$, if $m$ is the smallest rank
message it has encountered so far in the current phase.  
\begin{lemma}\label{lem:unique}
With probability at least $1-1/n^{k}$ for any fixed constant $k$, $O(\log n)$ sized rank messages are sufficient to ensure that there is at most one node $\ell$ that survives the first $D$ rounds of the phase and does not receive  a rank message smaller than its own rank in the current phase. 
\end{lemma}
\begin{proof} We will argue under the condition that no two rank messages are equal (which happens with probability at most $1/n^k$.
We only need to show that there cannot be two nodes $\ell_1$ and $\ell_2$ that
satisfy the condition for $\ell$ stated in the observation. Without loss of
generality, suppose $\ell_1$ generated the smaller random number and the
corresponding rank message was flooded. Since both nodes remain in the network
for $D$ rounds, node $\ell_2$ must have received that smaller rank message,
thus leading to a contradiction.
\end{proof}

If this node $\ell$ does not exist, then the phase has failed. If $\ell$ exists, then, by comparing its own rank message with the smallest rank message, it can clearly detect that it is $\ell$ and it elects itself the leader. In every subsequent round until $\ell$ is removed from the network,  $\ell$  sends out  a beep message containing its (unique) id and a timestamp  to indicate its presence in the network. The beep messages are conditionally flooded with the condition {\tt <current beep has the latest time stamp amongst received beeps and the time stamp is within last D rounds>}. Every other node therefore floods the latest beep message it has heard. Furthermore,  if a node receives a beep message with a timestamp within the last $D$ rounds, it sets its leader variable to $\ell$'s id. See Algorithm ~\ref{alg:election} for the detailed pseudo code.
\small
\begin{algorithm}[h]
\caption{Protocol executed by every node $u$ in a phase of the leader election process.}    
\label{alg:election}      
\begin{algorithmic} [1]  
\item[] \textbf{First $\boldmath D$ rounds of the phase:}
\STATE Let $p_u$ denote the number of phases that $u$ has been active in the current leader election process.
\STATE Generate a random number drawn from the exponential distribution with parameter $\lambda_u = 2^{p_u}$.
\STATE Broadcast a rank message containing $u$'s id and the generated random number.
\STATE During the first $D$ rounds, $u$ broadcasts the smallest rank message encountered so far. 
\item[]
\item[] \textbf{At the end of first  $\boldmath D$ rounds of the phase:}
\IF{the smallest rank message was generated by $u$}
\STATE Node $u$ concludes that it is the unique node  $\ell$ (cf.\ Lemma~\ref{lem:unique}) 
\STATE Set leader variable to self.
\STATE Generate and flood a \beep\ message every round under the condition that only the latest \beep\ message is propagated and \beep s older than $D$ rounds are discarded.
\STATE  Exit the leader election process and return to Algorithm~\ref{alg:frame}.

\ELSE
\STATE Node $u$ knows it is not $\ell$.
\ENDIF
\item[]
\item[] \textbf{Second $\boldmath D$ rounds of the phase:}
\STATE \COMMENT{The following code is executed only by non-$\ell$ nodes}
\FOR{each round of the second $\boldmath D$ rounds of the phase}
\IF {at least one \beep\ was received}
	\STATE $B \leftarrow$ the \beep\ with the latest time stamp among the set of \beep s received so far in this phase.
	\STATE \COMMENT{Since \beep s are designed to be discarded in $D$ rounds,  $B$ must have been generated in the current phase.}
		\STATE Set leader variable to the id contained in $B$. 		
		\STATE Broadcast $B$.
		\STATE  Exit the leader election process and return to Algorithm~\ref{alg:frame}.

\ENDIF
\ENDFOR
\end{algorithmic}
\end{algorithm}
\normalsize
\begin{lemma}\label{lem:correct}
  With high probability, Algorithm~\ref{alg:frame} guarantees agreement, validity, and stability.
\end{lemma}
\begin{proof}
Validity and stability conditions are quite straightforward. Validity holds because a node (that is not a leader) sets its leader variable only after hearing a beep message with timestamp in the last $D$ rounds. Stability, similarly, holds because  a node must have heard a beep timestamped within the last $D$ rounds if the node was still in the network. Therefore, when it has not heard a beep with a timestamp in the last $D$ rounds, it can safely conclude that the leader has left the network.

We now focus on the agreement condition. Suppose for the sake of contradiction that two nodes have their leader variables set to two different nodes $\ell_1$ and $\ell_2$.  They have both been leaders in the last $D$ rounds, which is not sufficient time for one leader to have left the network and another to have been elected leader. Therefore, they both considered themselves leaders in some round $r$. This implies that either (i) they both elected themselves leaders in the same phase (which is not possible as a consequence of Lemma~\ref{lem:unique}) or (ii) one of them (say $\ell_1$ without loss of generality) was elected leader in a phase prior to the other. Since $\ell_1$ is still in the network, it must have been sending beep messages every round since its election, which implies that another election process would not have started until round $r$. Thus $\ell_2$ could not have elected itself as leader. Thus no node can consider $\ell_2$ as leader without violating the validity constraint, thereby leading to a contradiction.
\end{proof}

We would now like to show that the termination time of any node $u$ that is without a leader is in $O(D \log n)$ with high probability of the form $1 - O(1/n)$. Towards this goal, we argue based on a node $u$ whose leader variable is set to $\bot$ either because it has just entered the network or has not heard a beep message in the last $D$ rounds. We now state several assumptions, all made without any loss of generality.  
\begin{compactenum}
\item We assume that $u$ is active; it only takes one phase for it to turn active. 

\item We assume  that we are at the start of the first phase in which $u$ is active and does not have a leader. 

\item If $u$ leaves the network within $O(D \log n)$ rounds, then the termination claim holds vacuously. Therefore, we assume $u$ remains in the network for at least $k D \log n$ rounds for some sufficiently large constant  $k>0$ that we will fix later. 
\item Finally, we assume that at time step $k D \log n$, no node has been active and without a leader for longer than $u$. If such a node existed, then we would base our analysis on that node rather than $u$.
\end{compactenum}
Our goal now is to show that, within these $k D \log n$ rounds, $u$  elects a leader with high probability.

For analysis purposes, color the nodes in the graph in round $1+(k-4) D \log n$ blue. In particular, $u$ is colored blue. The focus of our analysis will be the $2\log n$ phases from round $1+(k-4) D \log n$ to round $k D \log n$; we call these $2\log n$ phases {\em critical}.  Our next lemma shows that the probability with which a non-blue node produces the smallest random number during any one of these critical phases is very small. Before we state the lemma, we recall a fundamental property of the minimum value of $m$ exponential random numbers (cf. Chapter 8 in~\cite{MU05}). Let $X_1, X_2, \ldots, X_m$  be $m$ exponential random variables with parameters $\lambda_1, \lambda_2, \ldots, \lambda_m$, respectively. The probability that $X_i$, $i \in [m]$, is the smallest among all the $m$ random variables is given by:
\begin{equation} \label{eqn:exp}
\text{Pr}(\min\{X_1, X_2, \ldots, X_m\} \text{ is } X_i) = \frac{\lambda_i}{\lambda_1 + \lambda_2 + \cdots + \lambda_m}.
\end{equation}

\begin{lemma} \label{lem:non-blue}
Given that $k$ is a sufficiently large constant, the probability that a non-blue node draws the smallest random number in any one of the critical phases is at most $1/n^\gamma$, where $\gamma$ is a constant depending only on $k$.
\end{lemma}
\begin{proof}
Recall that we have assumed $u$ will remain in the network until round $k D \log n$. To maximize the probability that a non-blue node is chosen as leader, we must consider the case when all blue nodes other than $u$  are replaced by non-blue nodes starting from the first critical phase. Such non-blue nodes can generate exponential random numbers with parameters at most $2^{2\log n} = n^2$ during the $2\log n$ critical phases because they could have entered the network in the first critical phase or sometime thereafter.  Node $u$ on the other hand will generate exponential random numbers with a parameter no less than $2^{\frac{(k-4) D \log n}{2}} = n^{(k-4)/2}$. Using Equation~\ref{eqn:exp}, we see that a non-blue node $v$ will draw the smallest random number in a critical phase with probability at most
\[
\frac{n^2}{n^{(k-4)/2} + (n-1) n^2} < \frac{1}{n^{\frac{k-4}{2} - 1} }.
\]
Using the union bound over all non-blue nodes and over all critical phases, we can conclude that the probability that a non-blue node will draw the smallest random number in any of the critical phases is at most  $1/{(n^{\frac{k-4}{2} - 1})} \le 1/n^\gamma$
for any constant $\gamma$ provided $k$ is sufficiently large. In particular, when $k \ge 14$, the probability is at most $1/n$.
\end{proof}

In light of Lemma~\ref{lem:non-blue}, we condition the rest of our analysis on the event that a blue node generates the smallest random number in every critical phase --- an event that occurs with probability $(1 - \frac{1}{n})$ when $k \ge 14$. Recall that we defined a phase to be successful if a leader was elected in that phase and the leader remained in the network till the end of the phase; otherwise, the phase is said to be a failure. 
Our goal now is to show that regardless of the strategy employed by the adversary, with high probability, there will be a critical  phase that is successful.

\begin{lemma}
Regardless of the strategy employed by the oblivious adversary, at least $\log n$ (out of the $2 \log n$) critical  phases  are each successful with probability at least $\frac{1}{2}$.
\end{lemma}
\begin{proof}
  Define the following time-varying function \[\Psi(r)  = \frac{\sum_{b \in B(r)} 2^{p_b}}{2^{p_u}},\] where $B(r)$ is the set of blue nodes in round $r$ and $p_b$ is the number of phases that node $b$ has been active in the current leader election process at the start of round $r$. We now state three  properties of $\Psi$. 
\begin{compactenum}
\item Recall that no node has been active in the current leader election process for longer than $u$. Therefore, the value of $\Psi$ at the start of the first critical phase is no more than $n$. 
\item Since (i) the number of blue nodes cannot increase and (ii) $p_u \ge p_b$ for every $b \in B$ at all times in the critical phases, $\Psi$ decreases monotonically  with respect to time. 
\item However, since we have assumed that $u$ is in the network till the end of round $k D \log n$, $\Psi$ is always at least 1.
\end{compactenum}
Consider any particular critical phase. Let $B_{\text{start}}$ be the set of blue nodes in the start of the phase in consideration and $B_{\text{end}}$ be the set of blue nodes at the end of the phase. Let $\Psi_{\text{curr}}$ be the value of $\Psi$ at the start of  the current phase in consideration and $\Psi_{\text{next}}$ be the value of $\Psi$ at the start of the phase immediately succeeding the current phase in consideration.  Suppose the adversary enforces a strategy of removing nodes such that the probability of success is less than $\frac{1}{2}$.   Since we are conditioning on the event that only blue nodes generate the smallest random numbers, applying Equation~\ref{eqn:exp}, the probability that the phase in consideration is a success is given by
\begin{equation} \label{eqn:fail}
\frac{\sum_{b' \in B_{\text{end}}} 2^{p_{b'}}}{\sum_{b \in B_{\text{start}}} 2^{p_b} } 
< \frac{1}{2},
\end{equation}
where $p_b$ is the number of phases in which $b$ has been active in the current leader election process at the start of the phase in consideration.
\begin{align*}
\Psi_{\text{next}}
 &= \frac{ \sum_{b \in B_{\text{end}}} 2^{p_{b}+1}}{2^{p_u+1}} = \frac{ \sum_{b \in B_{\text{end}}} 2^{p_{b}}}{2^{p_u}}\\
&< \left (\frac{\sum_{b \in B_{\text{start}}2^{p_b}}}{2} \right ) \frac{1}{2^{p_u}} = \frac{\Psi_{\text{curr}}}{2}. 
&(\text{by Equation~} \ref{eqn:fail})
\end{align*}
Thus, $\Psi$ reduces by a factor of $\frac{1}{2}$ in each phase in which the adversary employs a strategy to reduce the probability of success to less than $\frac{1}{2}$. Therefore, it follows that  the number of critical  phases that are successful with probability at least $\frac{1}{2}$  is at least $\log n$. 
\end{proof}
Clearly, the probability that none of those critical phases with success probability at least $\frac{1}{2}$ will  be successful is at most $\frac{1}{n}$.  Removing the conditioning that only blue nodes generate the smallest random numbers and no two random numbers generate the same exponentially distributed random numbers, we get the following result.

\begin{theorem} \label{thm:main}
  Algorithm~\ref{alg:frame} solves dynamic leader election and terminates in $O(D \log n)$ rounds with  high probability.
\end{theorem}


\small
\bibliographystyle{eptcs}
\bibliography{papers,papers1,papers2,leader}
\end{document}